\documentclass[lettersize,journal]{IEEEtran}
\usepackage{amsmath,amsfonts}
\usepackage{algorithmic}
\usepackage{algorithm}
\usepackage{array}
\usepackage[caption=false,font=normalsize,labelfont=sf,textfont=sf]{subfig}
\usepackage{textcomp}
\usepackage{stfloats}
\usepackage{url}
\usepackage{verbatim}
\usepackage{graphicx}
\usepackage{cite}
\usepackage{cite}
\usepackage{graphicx}
\usepackage{amsmath}
\usepackage{amsfonts}
\usepackage{epsfig}
\usepackage{booktabs}
\usepackage{epstopdf}
\usepackage{extarrows}
\usepackage{url}
\usepackage{color}
\usepackage{amssymb}
\usepackage{hyperref}
\hypersetup{
	colorlinks = true,
	linkcolor  = black,
	citecolor  = black,
	urlcolor   = black
}
\usepackage{amssymb,amsthm}
\usepackage{graphicx}
\usepackage{flushend}

\makeatletter
\newcommand*\bigcdot{\mathpalette\bigcdot@{.5}}
\newcommand*\bigcdot@[2]{\mathbin{\vcenter{\hbox{\scalebox{#2}{$\m@th#1\bullet$}}}}}
\makeatother

\newtheorem{theorem}{Theorem}
\newtheorem{lemma}{Lemma}
\newtheorem{prop}{Proposition}

\begin{document}

\title{\LARGE{Secure Rate-Splitting Multiple Access Transmissions in LMS Systems }}

\author{Minjue He, Hui Zhao, Xiaqing Miao, Shuai Wang,  and Gaofeng Pan \vspace{-0.5cm}
\thanks{This work was supported by the National Key Research and Development Program of China under Grant 2021YFC3320200, and in part by the NSF of China under Grant 62171031 and 62101050. The associate editor coordinating the review of this paper and approving it for publication was G. Bacci. \emph{(Minjue He and Hui Zhao contributed equally to this work.)
(Corresponding author: Xiaqing Miao.)}}
\thanks{M. He is with Science and Technology on Electronic Information Control Laboratory, Chengdu 610043 China (e-mail: heminjue88@outlook.com).}
\thanks{H. Zhao is with the Communication Systems Department, EURECOM, 06410 Sophia Antipolis, France (e-mail:
hui.zhao@eurecom.fr).}
\thanks{X. Miao, S. Wang, and G. Pan are with the School of Information and Electronics, Beijing Institute of Technology, Beijing 100081, China (e-mail: xqmiao@bit.edu.cn; swang@bit.edu.cn; gfpan@bit.edu.cn).}
}

\markboth{IEEE Communications Letters}%
{Shell \MakeLowercase{\textit{et al.}}: A Sample Article Using IEEEtran.cls for IEEE Journals}


\maketitle

\begin{abstract}
This letter investigates the secure delivery performance of the rate-splitting multiple access scheme in land mobile satellite (LMS) systems, considering that the private messages intended by a terminal can be eavesdropped by any others from the broadcast signals. Specifically, the considered system has an $N$-antenna satellite and numerous single-antenna land users. Maximum ratio transmission (MRT) and matched-filtering (MF) precoding techniques are adopted at the satellite separately for the common messages (CMs) and for the private messages (PMs), which are both implemented based on the estimated LMS channels suffering from the Shadowed-Rician fading. Then, closed-form expressions are derived for the ergodic rates for decoding the CM, and for decoding the PM at the intended user respectively, and more importantly, we also derive the ergodic secrecy rate against eavesdropping. Finally, numerical results are provided to validate the correctness of the proposed analysis models, as well as to show some interesting comparisons.
\end{abstract}


\section{Introduction}
\IEEEPARstart{A}{s} land mobile satellite (LMS) systems are inherently typical multi-user systems, realizing robust multi-user access and interference management in LMS is vital for providing the expected quality of service \cite{PanTMC}.  Relying on the rate-splitting principle, rate-splitting multiple access (RSMA) has been recently proposed and regarded as a promising multiple access, interference management, and multi-user scheme in the broader communication community \cite{Bruno,Joudeh}. Therefore, considerable research has been conducted to study the performance of integrating RSMA with LMS systems regarding transmission sum rate \cite{LinIOT21,Khan,Huang}, max-min fairness \cite{YinTWC}, and unmet system capacity \cite{Cui}. Specifically, Refs. \cite{Khan, Huang,Cui,YinTWC} adopted RSMA for satellite-terrestrial communications, while RSMA was adopted for unmanned aerial vehicle-terrestrial transmission in \cite{LinIOT21}.

Information security issues can not be ignored in LMS systems. Because signals are delivered in such large-scale transmission space, it leads to the increased eavesdropping probability for malicious receivers located in the coverage space of the satellite. However, none of the aforementioned works has studied the secrecy performance of applying RSMA in satellite-terrestrial communication systems. Though a secure beamforming scheme was proposed for RSMA-based cognitive satellite-terrestrial networks in the presence of multiple eavesdroppers \cite{LinTWC21}, RSMA only worked in cognitive terrestrial systems, not in satellite-terrestrial transmissions.

To the best of the authors' knowledge, no studies have been found to investigate the secrecy performance of RSMA in LMS systems. This letter considers a typical LMS system consisting of a satellite delivering both the common messages (CMs) and the private messages (PMs) simultaneously to multiple terrestrial terminals under the RSMA scheme to improve transmission performance. It is obvious that each terrestrial terminal can eavesdrop on the PMs intended by any others from the received signal. In this letter, by integrating the analytical methods in \cite{Salem,QiZhang,Zhao_VectorCC,Yeon_Geun}, we derive closed-form expressions for the ergodic rates for decoding the CM and for decoding the PM at the intended user respectively, as well as derive the ergodic secrecy rate against eavesdropping\footnote{\emph{Notations.} We use $\mathbb{C}$ to denote the complex number set. For an integer $K>0$, $[K] \triangleq \{1,2,\cdots, K\}$. ${\bf I}_L$ is $L \times L$ identity matrix, while ${\bf 0}_L$ denotes the $L \times 1$  vector with all elements equaling zero. ${\rm Nakagami}(m,\Omega)$, $\mathcal{N}(0,\sigma^2)$ and $\mathcal{CN}(0,\sigma^2)$ denote the Nakagami-$m$ distribution with the shape parameter $m$ and the spread parameter $\Omega$, the normal distribution with zero-mean and the variance $\sigma^2$, the complex normal distribution with zero-mean and the variance $\sigma^2$ respectively. For a matrix ${\bf A}$, we use ${\bf A}^T$, ${\bf A}^*$ and ${\bf A}^H$ to denote the non-conjugate transpose, the conjugate part and the conjugate transpose of ${\bf A}$ respectively. $\mathbb{E}\{\cdot\}$ and ${\rm Tr}\{\cdot\}$ denote the average and trace operators respectively. $|\cdot|$ and $||\cdot||$ denote the magnitude of a complex number and the norm-2 operator respectively.}.


\section{System Model}
This work considers a typical LMS system consisting of a satellite transmitting information to multiple terrestrial terminals via the RSMA scheme and that each terrestrial terminal can eavesdrop on the PMs intended by any other users from the received signal. Some assumptions are made as follows.
\begin{itemize}
    \item An $N$-antenna satellite  communicates with $K$ single-antenna land users using RSMA technique.
     The maximum ratio transmission (MRT) is adopted for the CMs, denoted by $x_c$ with unit-power, while the matched-filtering (MF) precedes the PMs to the $K$ users ($x_1,x_2,\cdots,x_K$), where $x_k \ (k \in [K])$ with unit-power is the PM intended by the $k$-th user.
    \item The satellite-terrestrial channels are modelled as independent identically distributed (i.i.d) Shadowed-Rician (SR) fading channels \cite{Abdi}, whose statistics are determined by $m$, $b$, and $\Omega$. We use  ${\rm SR}(m,\Omega,b)$ to denote a specific SR fading channel. Specifically, the SR channel gain is $h = h^{(l)} + h^{(s)}$, where $h^{(l)} \sim {\rm Nakagami(m,\Omega)}$ models the line-of-sight component, and $h^{(s)} \sim \mathcal{CN}(0,2b)$ models the scatter component.
 \item The downlink channel training is perfect at the receivers, while the channel feedback to the satellite appears some errors, i.e., imperfect channel state information (CSI) at the transmitter \cite{Joudeh,Salem}.  For the maximum likelihood (ML) estimator in TDD training and for the actual channel vector ${\bf h} \in \mathbb{C}^{N \times 1}$, the ML estimate at the satellite is $\hat {\bf h} = {\bf h} + \tilde {\bf h}$, where $\tilde {\bf h} \sim \mathcal{CN}({\bf 0}_N, \sigma_e^2 {\bf I}_N)$ results from the additive white Gaussian noise (AWGN) during the CSI estimation, and which is independent of ${\bf h}$ \cite{Arti,Jinlong}.
\end{itemize}

Given the precoding matrix ${\bf W} = [{\bf w}_c, {\bf w}_1, {\bf w}_2, \cdots, {\bf w}_K] \in \mathbb{C}^{N \times (K+1)}$ and the RSMA power-splitting diagonal matrix ${\bf P} = {\rm Diag}(\sqrt{\rho}, \sqrt{\bar \rho}, \sqrt{\bar\rho}, \cdots, \sqrt{\bar\rho}) \in \mathbb{C}^{(K+1) \times (K+1)}$ where $\bar \rho \triangleq 1-\rho$ and $\rho \in [0,1]$ is responsible for power splitting between CMs and PMs, the transmitted signal is
\begin{align}\label{tranmit_signal_eq}
    {\bf s} = \sqrt{\alpha} {\bf W P x} = \sqrt{\alpha \rho} {\bf w}_c x_c + \sum\nolimits_{i=1}^K \sqrt{\alpha \bar\rho} {\bf w}_i x_i,
\end{align}
where $\alpha$ is responsible for the power normalization averaged over signal symbols and channel realizations under the average power constraint $P_t$. Specifically, $\alpha = \frac{P_t}{\mathbb{E}\{ {\rm Tr} \{ {\bf W}^H {\bf W} {\bf P}^2\}\}}$.\footnote{It is easy to check that
$\mathbb{E}\{ ||{\bf s}||^2\} \!=\! \alpha \mathbb{E}\{ {\bf x}^H {\bf P} {\bf W}^H {\bf W} {\bf P} {\bf x} \} \!=\! \alpha \mathbb{E}\{ {\rm Tr} \{  {\bf W}^H {\bf W} {\bf P} {\bf x} {\bf x}^H {\bf P} \} \} \!=\! \alpha {\rm Tr} \{ \mathbb{E}\{{\bf W}^H {\bf W}\} {\bf P} \mathbb{E}\{ {\bf x} {\bf x}^H \} {\bf P} \} = P_t$.}
The received signal at the $k$-th user can be expressed as
\begin{align}\label{y_k_eq_initial}
    y_k 
    &= \sqrt{\alpha \rho} {\bf h}_k^T {\bf w}_c x_c + \sum\nolimits_{i=1}^K \sqrt{\alpha \bar\rho} {\bf h}_k^T {\bf w}_i x_i + n_k,
\end{align}
where ${\bf h}_k \!\in\! \mathbb{C}^{N \times 1}$ denotes the channel vector from the satellite to the $k$-th user, and $n_k \!\sim\! \mathcal{CN}(0,\sigma_k^2)$ denotes the AWGN.

At the user side, after perfectly removing the CM $x_c$ using the Successive Interference Cancellation (SIC) technique \cite{Salem}, the received signal at the $k$-th user becomes
\begin{align}\label{y_k_private}
    y_k^{(p)} 
    = \sqrt{\alpha \bar \rho}{\bf h}_k^T {\bf W}^{(p)} {\bf x}^{(p)} + n_k,
\end{align}
where ${\bf x}^{(p)} \triangleq [x_1, x_2, \cdots, x_K] \in \mathbb{C}^{K \times 1}$ and ${\bf W}^{(p)} \triangleq   [{\bf w}_1, {\bf w}_2, \cdots, {\bf w}_K] \in \mathbb{C}^{N \times K}$.

\subsection{Performance Metrics}
The sum rate can be calculated by $R_{\rm sum} = R^{(c)} + \sum\nolimits_{k=1}^K R_k^{(p)}$, where $R^{(c)} = \min\big\{R_1^{(c)}, R_2^{(c)}, \cdots, R_K^{(c)}\big\}$ is the rate for the CM, $R_k^{(c)}$ is the rate for the CM at the $k$-th user, and $R_k^{(p)}$ is the rate for the PM intended by the $k$-th user.

As imperfect CSI is assumed, the delay-tolerant transmission is considered\footnote{In delivering videos from a geosynchronous (GEO) satellite, the land user can have a cache or buffer to store some video content for up to several minutes in order to have a seamless viewing when decoding the received signals from the GEO \cite{Zhao_ICC}. Therefore, the processing delay of advanced channel coding and signal regeneration in low SNR is acceptable in this case.}. Thus, sending the CMs
and the $k$th PM at ergodic rates given by $\bar R_k^{(c)} = \mathbb{E}\{ R_k^{(c)}\}$ and $\bar R_k^{(p)} = \mathbb{E}\{R_k^{(p)}\}$, respectively, guarantees successful
decoding by the $k$-th user \cite{Joudeh}. Accordingly, the ergodic sum rate can be evaluated by
\begin{align}
    \bar R_{\rm sum} = \mathbb{E}\{R_{\rm sum}\} = \min_{j \in [K]} \big\{ \bar R_j^{(c)} \big\} + \sum\nolimits_{k=1}^K \bar R_k^{(p)}.
\end{align}

In this work, we are also interested in the particular vulnerabilities of RSMA to eavesdropping. In this model, the eavesdropper can be any user, $i$, in the system trying to decode the $k$-th PM by $a)$ exploiting the CM together with the leakage caused by imperfect CSI, and $b)$ the leakage caused by the MF precoding\footnote{As indicated in \cite{Sadeghi}, ZF often outperforms MRT for multi-user multicasting in the massive MIMO regime that we are interested in. This leads to the fact that the optimal $\rho$ in RSMA often equals zero under MRT-ZF precoding in massive MIMO. It is worth noting that MRT-MMSE holds similar conclusions according to our simulation results, which are removed due to space limitation, and which will be deeply investigated in our future work. In contrast,  MRT-MF allows us to pick up a non-zero $\rho$ to achieve higher throughput than conventional MF (cf. Fig.~\ref{ER_rho_fig_ref}). Moreover,  MF is simple for implementation and its performance is as good as MMSE in low SNR \cite{QiZhang,Zhao_VectorCC,Yeon_Geun}. Thus,  MRT-MF suits the low-SNR-governed LMS system in the massive MIMO regime.  }. Therefore, the ergodic secrecy rate can be defined as $\bar R_k^{(s)} = \mathbb{E} \big\{ \big[ \big(R^{(c)} + R_k^{(p)}\big) - \max_{i\in [K], i \neq k} \big\{ R^{(c)} + R_{i \to k}^{(p)}\big\} \big]^+\big\}  = \mathbb{E} \big\{ \big[ R_k^{(p)} - \max_{i\in [K], i \neq k} \big\{  R_{i \to k}^{(p)}\big\} \big]^+\big\}$, where $[x]^+ \triangleq \max(x,0)$, and $R_{i \to k}^{(p)}$ is the rate at which user $i$ can decode the $k$-th PM. To ease the calculation, a tight approximation is proposed in \cite[Appendix F]{Salem}, as
\begin{align}\label{ESR_lower_def}
        \bar R_k^{(s)} &\approx \tilde R_k^{(s)} \triangleq
 \Big[ \bar R_k^{(p)} -  \max_{i\in [K], i \neq k} \Big\{ \tilde R_{i \to k}^{(p)}  \Big\}  \Big]^+ ,
    \end{align}
    where $\tilde R_{i \to k}^{(p)} \triangleq \log_2\big( 1 + \mathbb{E}\big\{\gamma_{i \to k}^{(p)} \big\} \big)$ is the upper-bound of $\bar R_{i \to k}^{(p)}$ after using the Jensen's inequality.

\subsection{MRT-MF Precoding}
MRT is implemented for the common stream, and MF precoding is applied for the private streams. Thus, the precoding vector for the CM is
${\bf w}_c = \sum\nolimits_{i=1}^K \hat {\bf h}^*_i$, where $\hat {\bf h}_i \in \mathbb{C}^{N \times 1}$ is the ML estimate of the actual channel vector ${\bf h}_i$. Specifically, the $\vartheta$-th element in $\hat{\bf h}_i$ (denoted by $\hat{\bf h}_i(\vartheta)$)  follows ${\rm SR}(m_i, \Omega_i, b_i+\sigma_{e_i}^2/2)$ for any $i \in [K]$ and $\vartheta \in [N]$.

The MF precoder based on the estimated channel is
${\bf W}^{(p)} = \big[ \hat{\bf h}_1^*, \hat{\bf h}_2^*, \cdots, \hat{\bf h}_K^* \big] \in \mathbb{C}^{N \times K}$. In the following, we will derive the power normalization factor $\alpha$ for the proposed precoding scheme. Before that, we define two parameters as follows
\begin{align}
   & B_i \triangleq \mathbb{E}\{|\hat {\bf h}_i(\vartheta)|^2\} = 2b_i+\sigma_{e_i}^2 + \Omega_i, \\
   &C_i \triangleq \mathbb{E}\{\hat {\bf h}_i(\vartheta)\} = \frac{\Gamma(m_i+\frac{1}{2})}{\Gamma(m_i)} \sqrt{\frac{\Omega_i}{m_i}}.
\end{align}
Now, we can present the expression for $\alpha$ in Proposition~\ref{alpha_prop}.

\begin{prop}\label{alpha_prop}
    Under the proposed MRT-MF precoding scheme designed for the estimated SR channel, the power normalization factor $\alpha$ in \eqref{tranmit_signal_eq} takes the form
    \begin{align}\label{alpha_eq}
    \alpha = \frac{P_t}{ N  \sum_{i=1}^K B_i + \rho N \sum_{i=1}^K \sum_{j=1,j\neq i}^K C_i C_j }.
\end{align}
\end{prop}

\begin{proof}
By considering the designed ${\bf w}_c$ and ${\bf W}^{(p)}$ and after some simple mathematical manipulations, we can have that
    \begin{align}\label{E_WW_eq}
    \mathbb{E}\{ {\rm Tr} \{ {\bf W}^H {\bf W} {\bf P}^2\}\}
     = \rho \mathbb{E}\big\{ ||{\bf w}_c||^2 \big\}+ \bar\rho N  \sum\nolimits_{i=1}^K B_i.
\end{align}
As for $\mathbb{E}\big\{ ||{\bf w}_c||^2 \big\}$, we can further have that
\begin{align}\label{E_wc_eq}
     \mathbb{E}&\big\{ ||{\bf w}_c||^2 \big\} = \mathbb{E}\Big\{ \sum\nolimits_{j=1}^K \hat {\bf h}_j^T \sum\nolimits_{i=1}^K \hat {\bf h}_i^* \Big\} \notag\\
    &= \sum\nolimits_{i=1}^K \mathbb{E} \big\{ ||\hat {\bf h}_i^*||^2 \big\} + \sum\nolimits_{i=1}^K \sum\nolimits_{j=1,j\neq i}^K \mathbb{E} \big\{ \hat {\bf h}_j^T \big\}  \mathbb{E} \big\{ \hat {\bf h}_i^* \big\} \notag\\
    &\overset{(a)}{=} N \sum\nolimits_{i=1}^K B_i +  N \sum\nolimits_{i=1}^K \sum\nolimits_{j=1,j\neq i}^K C_i C_j,
\end{align}
where $(a)$ follows from considering the definitions of $B_i$ and $C_i$. Combining \eqref{E_WW_eq} and \eqref{E_wc_eq} finally yields \eqref{alpha_eq}, which concludes the proof.
\end{proof}

\section{Performance Analysis}
Before presenting the main results, we define some parameters for any $k,i,j \in [K]$ and $k \neq i \neq j$, as follows
\begin{align}
    &D_k \triangleq \mathbb{E}\big\{ ||{\bf h}_k||^4 \big\}, \  E_{k,i} \triangleq \mathbb{E} \big\{ | {\bf h}_k^T {\bf h}_i^*|^2 \big\}, \notag \\
    &F_{i,k} \triangleq \mathbb{E} \big\{ {\bf h}_i^T   {\bf h}_k^*  ||{\bf h}_k||^2 \big\}, \
    G_{i,k,j} \triangleq \mathbb{E} \big\{ {\bf h}_i^T  {\bf h}_k^* {\bf h}_k^T {\bf h}_j^* \big\},
\end{align}
and the expressions for them are derived in Lemma~\ref{DEF_prop}.

\begin{lemma}\label{DEF_prop}
    The closed-form expressions for $D_k$, $E_{k,i}$, $F_{i,k}$ and $G_{i,k,j}$ respectively take the forms
    \begin{align}
        &D_k = N \Big( 4b_k^2 + 4 b_k \Omega_k + \frac{\Omega_k^2}{m_k} \Big) \!+\! N^2 A_k^2, \notag\\
        &E_{k,i} = N  A_k A_i + N(N-1) C_k^2  C_i^2, \notag\\
        &F_{i,k} \!=\!  N C_i \bigg(   \frac{\Gamma(m_k+{3}/{2})}{\Gamma(m_k)} \Big( \frac{\Omega_k}{m_k} \Big)^{\frac{3}{2}} \!+\! 3 b_k C_k  \!+\! C_k (N-1) A_k \!\! \bigg), \notag\\
        &G_{i,k,j} = N C_i C_j \big( A_k + (N-1) C_k^2  \big), \notag
    \end{align}
 where $A_k \!\triangleq\! \mathbb{E}\{ |{\bf h}_k(\vartheta)|^2 \} \!=\! 2b_k + \Omega_k$ for $k \in [K]$ and $\vartheta \in [N]$.
\end{lemma}

\begin{proof}
    See Appendix I.
\end{proof}

\subsection{Ergodic Rate for Decoding CMs}
By using the relationship between the actual channel and the estimated channel, we can derive the received signal for decoding CMs in \eqref{y_k_eq_initial} under the proposed MRT-MF precoding scheme as \eqref{y_k_eq_MRT_MF}, shown at the top of this page,
\begin{figure*}
    \begin{align}\label{y_k_eq_MRT_MF}
        y_k \!=\! \sqrt{\alpha \rho}   {\bf h}_k^T \Big( \sum\nolimits_{i=1}^K  {\bf h}^*_i \Big) x_c \!+\! \sqrt{\alpha \rho}  {\bf h}_k^T \Big( \sum\nolimits_{i=1}^K \tilde {\bf h}^*_i \Big) x_c   \! +\!  \sum\nolimits_{i=1}^K \sqrt{\alpha  \bar\rho}  {\bf h}_k^T {\bf h}_i^* x_i  + \sum\nolimits_{i=1}^K \sqrt{\alpha \bar\rho} {\bf h}_k^T \tilde  {\bf h}_i^* x_i + n_k.
    \end{align}
\rule{18cm}{0.01cm}
\end{figure*}
and therefore, the SINR for decoding $x_c$ at the $k$-th user takes the form
  \begin{align}\label{SINR_k_c_eq}
    \gamma_k^{(c)} = \frac{\alpha \rho \big|{\bf h}_k^T \big( \sum\nolimits_{i=1}^K {\bf h}^*_i \big) \big|^2 + \alpha \rho    \sum\nolimits_{i=1}^K \sigma_{e_i}^2 ||{\bf h}_k||^2  }{\sigma_k^2  + \alpha  \bar\rho\sum\nolimits_{i=1}^K  \big( \sigma_{e_i}^2 ||{\bf h}_k||^2 +  | {\bf h}_k^T  {\bf h}_i^*|^2 \big)}.
\end{align}

Now, we present the first main result for the ergodic rate for decoding the CMs in RSMA-aided LMS systems.
\begin{theorem}\label{E_Rc_Thm}
    Under the proposed MRT-MF precoding scheme based on the estimated satellite-terrestrial channel, the ergodic rate for decoding the CM at the $k$-th user takes the form in \eqref{R_k_common_eq}, as shown at the top of this page,  where $\overset{\bigcdot}{\approx}$ represents that the derived approximation converges to $\bar R_{k}^{(c)}$ as $N \to \infty$ and $K \to \infty$.
\begin{figure*}
\begin{align}\label{R_k_common_eq}
    \bar R_k^{(c)} \overset{\bigcdot}{\approx} \tilde R_k^{(c)} \triangleq
    \log_2 \left( 1 + \alpha \rho \frac{ D_k +   \sum\nolimits_{\substack{i=1, i \neq k}}^K \big( 2F_{i,k} + E_{k,i} \big) +
   \sum\nolimits_{\substack{i=1, i \neq k}}^K \sum\nolimits_{\substack{j=1, j \neq k,i}}^K  G_{i,k,j} + N A_k \sum_{i=1}^K \sigma_{e_i}^2  }{\sigma_k^2  + \alpha \bar\rho \Big(D_k+ \sum_{i=1}^K \sigma_{e_i}^2 N A_k    +  \sum_{i=1,i\neq k}^K E_{k,i} \Big) } \right).
\end{align}
\rule{18cm}{0.01cm}
\end{figure*}
\end{theorem}

\begin{proof}
    Based on the SINR in \eqref{SINR_k_c_eq},  the ergodic rate for decoding the CMs at the $k$-th user takes the form $ \bar R_k^{(c)} = \mathbb{E}\big\{\log_2 \big( 1 + \gamma_k^{(c)} \big) \big\}$. According to \cite[Lem. 1]{QiZhang}, we can approximate $\bar R_k^{(c)}$ as
    \begin{small}
    \begin{align}\label{R_k_c_initial}
    \bar R_k^{(c)}  \overset{\bigcdot}{\approx}  \log_2 \! \left( \! 1+ \frac{  \alpha \rho \mathbb{E}\bigg\{ \Big|{\bf h}_k^T  \sum\limits_{i=1}^K {\bf h}^*_i  \Big|^2 +     \sum\limits_{i=1}^K \sigma_{e_i}^2 ||{\bf h}_k||^2 \bigg\}  }{\sigma_k^2  + \alpha  \bar\rho\sum\limits_{i=1}^K  \mathbb{E}\Big\{ \sigma_{e_i}^2 ||{\bf h}_k||^2 +  | {\bf h}_k^T  {\bf h}_i^*|^2 \Big\} } \! \right),
\end{align}
\end{small}
where the gap between $\bar R_{k}^{(c)}$ and $\tilde  R_{k}^{(c)}$ decreases as $N$  and $K$ increases \cite{QiZhang,Zhao_VectorCC}. For the useful signal power for decoding in \eqref{R_k_c_initial}, we can have that
\begin{align}
    &\mathbb{E}\Big\{ \big| {\bf h}_k^T  \sum\nolimits_{i=1}^K  {\bf h}^*_i  \big|^2 \Big\}
    = \mathbb{E}\Big\{ \big| {\bf h}_k^T   {\bf h}^*_k + \sum\nolimits_{i=1, i \neq k}^K {\bf h}_k^T  {\bf h}^*_i  \big|^2 \Big\}
  \notag\\
   & =  \mathbb{E} \big\{ || {\bf h}_k||^4 \big\}  +
   \sum\nolimits_{\substack{i=1\\ i \neq k}}^K \sum\nolimits_{\substack{j=1 \\ j \neq k,i}}^K  \mathbb{E} \big\{ {\bf h}_i^T  {\bf h}_k^* {\bf h}_k^T {\bf h}_j^* \big\}
   \notag\\
   &\hspace{0.5cm}+  2 \sum\nolimits_{\substack{i=1\\ i \neq k}}^K \mathbb{E} \big\{  {\bf h}_i^T {\bf h}_k^* ||{\bf h}_k||^2 \big\}
   + \sum\nolimits_{\substack{i=1\\ i \neq k}}^K
    \mathbb{E} \big\{ {\bf h}_i^T  {\bf h}_k^* {\bf h}_k^T {\bf h}_i^* \big\}.\notag
\end{align}
We can easily derive \eqref{R_k_common_eq} by using Lemma~\ref{DEF_prop} into the equation above and into \eqref{R_k_c_initial}, which concludes the proof.
\end{proof}

\subsection{Ergodic Rate for Decoding PMs}
After removing the CM via SIC, the received signal at the $k$-th user in \eqref{y_k_eq_MRT_MF} becomes $y_k^{(p)} = \sqrt{\alpha \bar\rho}  {\bf h}_k^T ({\bf h}_k^* + \tilde{\bf h}_k^* )x_k  + n_k + \sqrt{\alpha \bar\rho}\sum\nolimits_{i=1, i \neq k}^K   {\bf h}_k^T ({\bf h}_i^*+\tilde {\bf h}_i^*) x_i$.
The SINR for decoding $x_k$ at the $k$-th user takes the form \cite{Salem,Joudeh}
\begin{align}\label{SINR_kp_eq}
    \gamma_{k}^{(p)} &= \frac{\alpha \bar \rho || {\bf h}_k||^4 + \alpha  \bar\rho ||{\bf h}_k||^2  \sigma_{e_k}^2 }{\sigma_k^2  + \alpha  \bar\rho\sum_{\substack{i=1, i \neq k}}^K \big( |{\bf h}_k^T {\bf h}_i^*|^2 + ||{\bf h}_k||^2 \sigma_{e_i}^2 \big)}.
\end{align}

In Theorem~\ref{R_k_p_Thm}, we derive a tight approximation for the ergodic rate for decoding the PM intended by the $k$-th user.

\begin{theorem}\label{R_k_p_Thm}
    Under the proposed MRT-MF precoding scheme, a tight approximation for $\bar R_k^{(p)}$ takes the form $\bar R_{k}^{(p)} \overset{\bigcdot}{\approx} \tilde R_{k}^{(p)}$, where $\tilde R_{k}^{(p)}$ is given by
    \begin{align}\label{E_Rk_p_eq}
        \tilde R_{k}^{(p)} \triangleq \! \log_2 \! \Bigg( 1 \!+\!   \frac{\alpha \bar\rho D_k + \alpha \bar\rho \sigma_{e_k}^2 N A_k}{\sigma_k^2  \!+\! \alpha \bar\rho \sum\nolimits_{\substack{i=1 \\ i \neq k}}^K \big(  E_{k,i} + \sigma_{e_i}^2 N A_k \big) } \Bigg),
    \end{align}
    and where $\overset{\bigcdot}{\approx}$ represents that the derived approximation converges to $\bar R_{k}^{(p)}$ as $N \to \infty$ and $K \to \infty$, and where $A_k$, $D_k$ and $E_{k,i}$  are respectively given in Lemma~\ref{DEF_prop}.
\end{theorem}

\begin{proof}
The proof is similar to that of Theorem \ref{E_Rc_Thm}.
\end{proof}

\subsection{Ergodic Secrecy Rate}
Based on the RSMA principle, user $i$ detects first the CM and his own PM, then removes them using SIC to eavesdrop on the $k$-th PM. Accordingly, the received signal at user $i$
to detect the $k$-th PM is $y_{i \to k}^{(p)}  = \sqrt{\alpha \bar\rho}  {\bf h}_i^T ({\bf h}_k^* + \tilde{\bf h}_k^* ) x_k  +\sqrt{\alpha \bar\rho}\sum\nolimits_{j=1, j \neq k,i}^K   {\bf h}_i^T ({\bf h}_j^* + \tilde{\bf h}_j^* ) x_j  + n_i$. Then, the SINR for decoding $x_k$ at the $i$-th user is
\begin{align}\label{gamma_ik_eq}
\gamma_{i \to k}^{(p)} \!=\! \frac{\alpha \bar\rho \big( | {\bf h}_i^T {\bf h}_k^*|^2  + \sigma_{e_k}^2 ||{\bf h}_i||^2 \big)}{\sigma_i^2 + \alpha \bar\rho  \sum\nolimits_{j=1, j \neq k,i}^K \big( |{\bf h}_i^T {\bf h}_j^*|^2 + \sigma_{e_j}^2 ||{\bf h}_i||^2 \big)}.
\end{align}

\begin{prop}\label{R_i_k_prop}
    In low SNR, we can tightly approximate the upper-bound $\tilde R_{i \to k}^{(p)}$ in \eqref{ESR_lower_def}  as\footnote{It is intractable to derive an exact expression for $\tilde R_{i \to k}^{(p)}$ from which we can draw some useful insights. For that, we are interested in a simple but tight approximation in low SNR by following the proposed method in \cite{Yeon_Geun}.}
    \begin{align}\label{E_R_ik_eq}
        &\tilde R_{i \to k}^{(p)}  \approx \! \log_2 \! \Bigg( 1 \!+\! \frac{\alpha \bar \rho \big( E_{i,k}+\sigma_{e_k}^2 N A_i \big)}{\sigma_i^2 + \alpha \bar \rho \sum_{j=1,j\neq k,i}^K \big( E_{i,j} + \sigma_{e_j}^2 N A_i \big) }  \Bigg).
    \end{align}
\end{prop}

\begin{proof}
As proposed in \cite{Yeon_Geun}, for a random variable $X$ with a finite variance ${\rm Var}\{X\}$, we have that $Y = c X$ is \emph{near deterministic} (``nearly all" random variables close to the mean) as $c$ goes to zero. This is can be easily verified by using the Chebyshev’s inequality. For an arbitrary $\epsilon >0$, we have that ${\rm Pr} \! \big\{ \big| Y - \mathbb{E}\{Y\}  \big| \ge \epsilon  \big\} \! \le \! \frac{ c^2 {\rm Var}\{X\}}{\epsilon^2} \longrightarrow 0, \text{ as } c \to 0$.
Therefore, for a fixed $\epsilon>0$, $Y$ converges to $\mathbb{E}\{Y\}$ as  $c \to 0$ \cite{Yeon_Geun}. By using this approximation method, for $\alpha \to 0$ (i.e., low SNR) in \eqref{gamma_ik_eq}, we have that
\begin{align}\label{h_ik_near}
    &\alpha \big| {\bf h}_i^T {\bf h}_k^* \big|^2 \approx  \alpha \mathbb{E} \big\{ \big|  {\bf h}_i^T {\bf h}_k^* \big|^2 \big\}
    = \alpha E_{i,k}, \\
    &\alpha \big|{\bf h}_i^T {\bf h}_j^* \big|^2 \approx  \alpha \mathbb{E} \big\{ \big|  {\bf h}_i^T {\bf h}_j^* \big|^2 \big\}
    = \alpha E_{i,j},\\
    &\alpha || {\bf h}_i||^2 \approx \alpha \mathbb{E} \{ || {\bf h}_i||^2 \} =  \alpha N A_i.\label{h_j_near}
\end{align}
Applying \eqref{h_ik_near}--\eqref{h_j_near} in \eqref{gamma_ik_eq} and then considering
that $\tilde R_{i \to k}^{(p)} = \log_2 \big( 1 \!+\! \mathbb{E} \big\{ \gamma_{i \to k}^{(p)} \big\} \big)$, finally yields \eqref{E_R_ik_eq} .
\end{proof}

\begin{theorem}
    Under the proposed MRT-MF precoding scheme and in low SNR, the tight bound in \eqref{ESR_lower_def} for the ergodic secrecy rate for the $k$-th user against eavesdropping from other users can be approximated as
    \begin{align}\label{ESR_k_eq}
        \tilde R_{k}^{(s)}
        \approx \Big[ \tilde R^{(p)}_k - \max_{i \in [K], i \neq k} \Big\{ \tilde R_{ i \to k}^{(p)} \Big\} \Big]^+,
    \end{align}
where $\tilde R_{k}^{(p)}$ and $\tilde R_{i \to k}^{(p)}$ are given in \eqref{E_Rk_p_eq} and \eqref{E_R_ik_eq} respectively.
\end{theorem}

\begin{proof}
    We can easily obtain \eqref{ESR_k_eq} by using Theorem~\ref{R_k_p_Thm} and Proposition~\ref{R_i_k_prop} in \eqref{ESR_lower_def}.
\end{proof}

\section{Numerical Results}
We proceed to numerically demonstrate the accuracy of the derived expressions.  To ease the simulation, we assume that the land users are statistically symmetric, i.e., $\sigma_k^2 = \sigma^2$, $m_k = m$, $b_k = b$ and $\Omega_k = \Omega$ for any $k \in [K]$. In the imperfect CSI model, we consider that the CSI quality is allowed to be scaled with the $\text{SNR} \triangleq P_t/\sigma^2$, where $\sigma^2$ is normalized to 1 for simplicity. For that, $\sigma_{e_k}^2 = ( \text{SNR} \cdot L)^{-1}$, as suggested in \cite{Jinlong}, where $L$ is the number of training symbols.  In the numerical results, we consider three typical shadowing scenarios in the LMS channels \cite{Abdi}, as listed in the following table.

\renewcommand*{\arraystretch}{1.35}%
\begin{table}[h!]
  \centering~
   \scalebox{0.9}{
  \begin{tabular}{ c| c c c }
  Shadowing Scenarios  & $m$ & $b$ & $\Omega$ \\
  \hline  		
  Frequent Heavy Shadowing (FHS) & 0.739 & 0.063 & $8.97 \times 10^{-4}$ \\
  Overall Results (ORs) & 5.21 & 0.251 & 0.278 \\
  Average Shadowing (AS) & \;\;10.1\; & 0.126 & 0.835 \\
  \end{tabular}}~\vspace{-2ex}
\end{table}
\renewcommand*{\arraystretch}{1}%

In Fig.~\ref{EC_Pt_fig_ref}, we plot the ergodic sum rate versus SNR with different CSI estimation accuracies in AS. It is easy to see that the ergodic sum rate becomes larger as $L$ increases in the low to medium SNR region, while the rates converge in the high SNR region. This is because the estimation error variance $\sigma_{e_k}^2 = (\text{SNR}\cdot L)^{-1}$ decreases as SNR increases, which finally leads to a marginal impact of estimation errors on the transmission performance when the SNR is sufficiently large.

\begin{figure}[!t]
             \centering
             \includegraphics[width= 3.5 in]{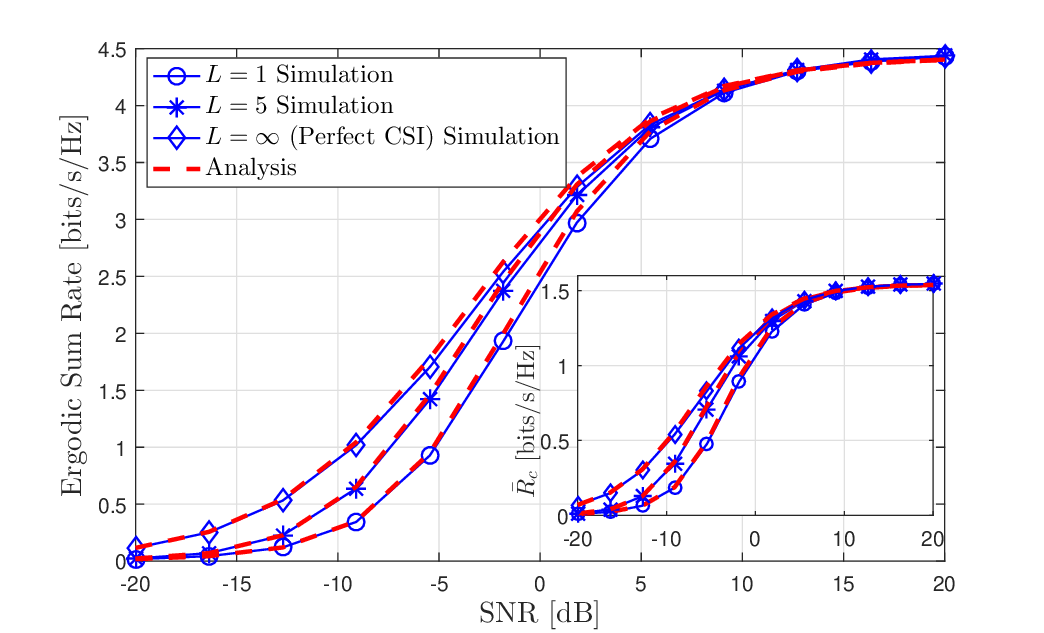}
             \vspace{-0.7cm}
            \caption{$\bar R_{\rm sum}$ and $\bar R_c$  vs. SNR for $N=8$, $K=2$ and $\rho =0.5$ in AS }\label{EC_Pt_fig_ref}\vspace{-0.2cm}
        \end{figure}	

Fig.~\ref{ESR_rho_fig_ref} plots the ergodic secrecy rate versus SNR for serving different numbers of land users in FHS. It is obvious that the secure performance is improved as $K$ decreases due to fewer potential eavesdroppers. In Figs.~\ref{EC_Pt_fig_ref}--\ref{ESR_rho_fig_ref}, both the ergodic sum rate and the ergodic secrecy rate are monotonically increasing with SNR, but they almost remain unchanged when the SNR is sufficiently large, as the MF precoder is interference-limited.

\begin{figure}[!t]
             \centering
             \includegraphics[width= 3.5 in]{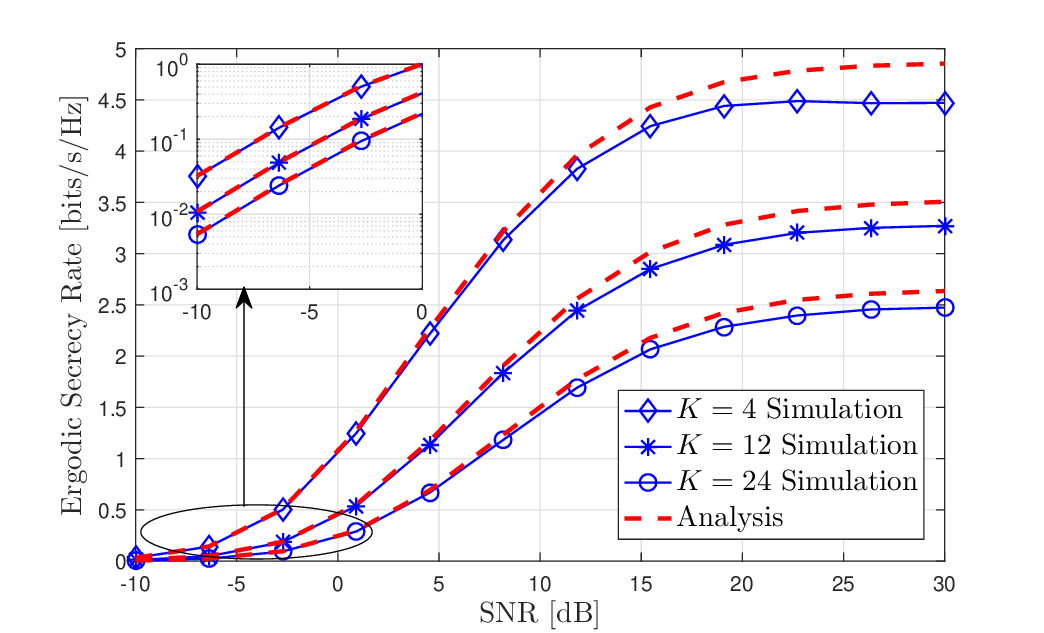}
             \vspace{-0.7cm}
            \caption{$\bar R_k^{(s)}$ vs. SNR for $N=128$, $\rho=0.5$ and $L=10$ in FHS}\label{ESR_rho_fig_ref}\vspace{-0.2cm}
        \end{figure}	

In RSMA, we are particularly interested in how the power-splitting factor $\rho$ affects both the ergodic sum rate and the ergodic secrecy rate. Fig.~\ref{ER_rho_fig_ref} presents not only the ergodic sum rate versus $\rho$ but also how the ergodic secrecy rate changes with $\rho$ in ORs.
In Fig.~\ref{ER_rho_fig_ref}, the ergodic sum rate first increases up to a peak point, after which it decreases as $\rho$ increases. In contrast, the ergodic secrecy rate is always monotonically decreasing with $\rho$. It is also worth noting that apart from boosting both the ergodic sum rate and the ergodic secrecy rate, a larger scale of antennas also makes the ergodic secrecy rate decrease slowly. This implies that multiple antennas play an important role in improving the tradeoff between the ergodic sum rate and the ergodic secrecy rate in RSMA.

Last but not least, from the presented numerical results, we can easily observe that the derived expressions provide good approximations of the real performance in the considered RSMA-aided LMS system, even for not-so-large $K$ and $N$.

\begin{figure}[!t]
             \centering
             \includegraphics[width= 3.5in]{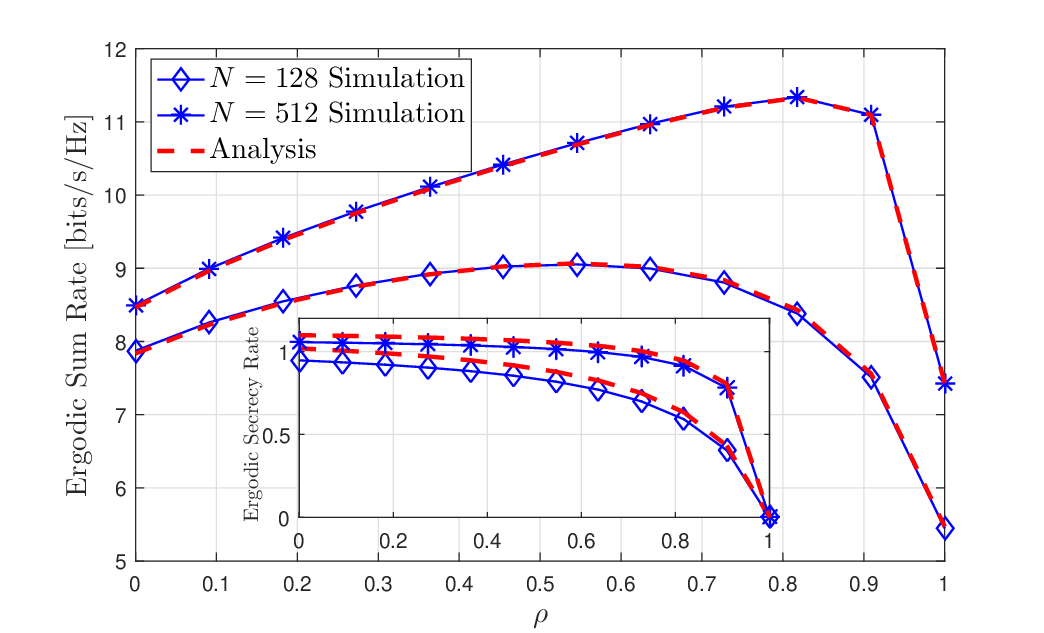}
             \vspace{-0.7cm}
            \caption{$\bar R_{\rm sum}$ and $\bar R_k^{(s)}$ vs. $\rho$ for $K=6$,  $\text{SNR}=0$ dB and $L=10$ in ORs}\label{ER_rho_fig_ref}\vspace{-0.2cm}
        \end{figure}

\section{Conclusions}
We have investigated the delivery performance of RSMA in LMS systems. Specifically, we have derived simple but tight approximations for the ergodic rates for decoding the CM and for decoding the PM at the intended user respectively. More importantly, we have also analyzed the secrecy performance of the considered RSMA-aided LMS system by deriving a tight approximation for the secrecy ergodic rate. 

\ifCLASSOPTIONcaptionsoff
  \newpage
\fi

\bibliographystyle{IEEEtran}				
\bibliography{IEEEabrv,aBiblio}


\appendices
\renewcommand{\thesectiondis}[2]{\Roman{section}:}
\section{Proof of Lemma~\ref{DEF_prop}}\label{DEF_prop_proof}
As ${\bf h}_k \sim {\rm SR}(m_k,\Omega_k,b_k)$, we can have that
\begin{align}
    &D_k = \mathbb{E}\big\{ || {\bf h}_k||^4 \big\} =  \mathbb{E}\bigg\{ \sum_{i=1}^N |{\bf h}_k(i)|^2 \sum_{j=1}^N | {\bf h}_k(j)|^2  \bigg\} \notag\\
    &= \sum_{i=1}^N  \mathbb{E} \big\{ | {\bf h}_k(i)|^4 \big\} \!+\! \sum_{i=1}^N \sum_{j=1, j\neq i}^N \!\! \mathbb{E} \big\{ | {\bf h}_k(i)|^2\big\} \mathbb{E} \big\{ | {\bf h}_k(j)|^2\big\}, \notag
\end{align}
which yields $D_k$ in Lemma~\ref{DEF_prop} by using \cite[Prop. 5.1]{Zhao_defense}.

For  $E_{k,i}$, we have that
$
E_{k,i} 
\!=\! \mathbb{E}\big\{ {\rm Tr} \{  {\bf h}_k^*  {\bf h}_k^T  {\bf h}_i^* {\bf h}_i^T\} \big\} \!=\! {\rm Tr} \big\{ \mathbb{E} \{  {\bf h}_k^*  {\bf h}_k^T \}  \mathbb{E} \{  {\bf h}_i^*  {\bf h}_i^T\} \big\},
$
which yields $E_{k,i}$ in Lemma~\ref{DEF_prop}.

For $F_{i,k}$, we have that
\begin{align}\label{G_j_k_1}
   F_{i,k} 
   &= \sum\nolimits_{\vartheta=1}^N \mathbb{E} \{ {\bf h}_i(\vartheta) \} \mathbb{E} \bigg\{  {\bf h}_k^*(\vartheta) \sum\nolimits_{\ell=1}^N  | {\bf h}_k(\ell)|^2 \bigg\},
\end{align}
where
\begin{align}
    &\mathbb{E} \bigg\{  {\bf h}_k^*(\vartheta) \sum\nolimits_{\ell=1}^N  | {\bf h}_k(\ell)|^2 \bigg\} \notag\\
    &= \mathbb{E} \{  {\bf h}_k^*(\vartheta) | {\bf h}_k(\vartheta)|^2 \} + \mathbb{E} \bigg\{  {\bf h}_k^*(\vartheta)  \sum\nolimits_{\ell=1, \ell \neq \vartheta}^N  | {\bf h}_k(\ell)|^2 \bigg\} \notag\\
    & = \mathbb{E} \{ {\bf h}_k^*(\vartheta) |{\bf h}_k(\vartheta)|^2 \}  + C_k (N-1)A_k.
\end{align}
We can write ${\bf h}_k(\vartheta)$ as ${\bf h}_k(\vartheta) = Z + X + \jmath Y$, where $Z, X, Y$ are independent. Specifically, $Z \sim {\rm Nakagami}(m_k,\Omega_k)$, while $X, Y \sim \mathcal{N}(0,b_k)$. We have that
\begin{align}\label{G_j_k_3}
    \mathbb{E} \{ {\bf h}_k^*(\vartheta) |{\bf h}_k(\vartheta)|^2 \} &= \mathbb{E} \{ (Z+X-\jmath Y) ((Z+X)^2+Y^2) \} \notag\\
    & = \mathbb{E} \{ Z^3 + 3 X^2 Z \} + \mathbb{E} \{ Z+X \} \mathbb{E} \{ Y^2 \}   \notag\\
    &= \frac{\Gamma(m_k+\frac{3}{2})}{\Gamma(m_k)} \Big( \frac{\Omega_k}{m_k} \Big)^{3/2} + 3 b_k C_k .
\end{align}
Combining \eqref{G_j_k_1}--\eqref{G_j_k_3} yields $F_{i,k}$ in Lemma~\ref{DEF_prop}.

Finally,
$
G_{i,k,j} = {\rm Tr} \big\{ \mathbb{E} \big\{ {\bf h}_k^* {\bf h}_k^T    \big\}  \mathbb{E} \big\{ {\bf h}_j^* \big\} \mathbb{E} \big\{ {\bf h}_i^T\big\} \big\},
$
which easily yields $G_{i,k,j}$ in Lemma~\ref{DEF_prop}.

\vfill

\end{document}